\documentclass[final,12pt]{article}
\usepackage{hyperref}
\usepackage{amssymb,amsmath,amsthm}
\usepackage{tikz}
\usetikzlibrary{arrows,backgrounds,decorations.pathmorphing,decorations.pathreplacing,positioning,fit}
\usepackage{enumerate}
\usepackage[conditional,light,first,bottomafter]{draftcopy}
\draftcopyName{DRAFT\space\today}{130}
\draftcopySetScale{65}
\usepackage[letterpaper,hmargin=3.3cm,vmargin={2.6cm,3.2cm}]{geometry}
\usepackage{setspace}
\makeatletter
\renewcommand{\section}{\@startsection%
{section}%
{1}%
{0em}%
{1.7em}%
{1.2em}%
{\normalfont\large\centering\bfseries}}
\renewcommand{\@seccntformat}[1]%
{\csname the#1\endcsname.\hspace{0.5em}}
\makeatother

\numberwithin{equation}{section}

\newtheorem{theorem}{Theorem}[section]
\newtheorem{proposition}{Proposition}[section]
\newtheorem{lemma}{Lemma}[section]
\newtheorem{corollary}{Corollary}[section]
\theoremstyle{definition}
\newtheorem{definition}{Definition}
\newtheorem{remark}{Remark}

\newtheorem*{convention}{Convention}

\newtheorem*{acknowledgments}{Acknowledgments}
\newcommand{\reals}{\mathbb{R}}
\newcommand{\nats}{\mathbb{N}}

\newcommand{\complex}{\mathbb{C}}
\newcommand{\abs}[1]{\left|#1\right|}

\newcommand{\inner}[2]{\left\langle#1,#2\right\rangle}

\newcommand{\eval}[1]{\upharpoonright_{#1}}
\def\cprime{$'$}
\def\ocirc#1{\ifmmode\setbox0=\hbox{$#1$}\dimen0=\ht0 \advance\dimen0
  by1pt\rlap{\hbox to\wd0{\hss\raise\dimen0
  \hbox{\hskip.2em$\scriptscriptstyle\circ$}\hss}}#1\else {\accent"17 #1}\fi}

\DeclareMathOperator{\im}{Im}
\DeclareMathOperator{\dom}{dom}

\DeclareMathOperator{\diag}{diag}
\DeclareMathOperator*{\res}{Res}

\DeclareMathOperator{\Span}{span}

\begin{document}
\begin{titlepage}
\title{Inverse problems for Jacobi operators II:\\
Mass perturbations of semi-infinite mass-spring systems
\footnotetext{%
Mathematics Subject Classification(2010):
34K29,  
47A75, 
47B36, 
70F17, 
}
\footnotetext{%
Keywords:
Infinite mass-spring system;
Jacobi matrices;
Two-spectra inverse problem.
}
}
\author{
\textbf{Rafael del Rio}
\\
\small Departamento de M\'{e}todos Matem\'{a}ticos y Num\'{e}ricos\\[-1.6mm]
\small Instituto de Investigaciones en Matem\'aticas Aplicadas y en Sistemas\\[-1.6mm]
\small Universidad Nacional Aut\'onoma de M\'exico\\[-1.6mm]
\small C.P. 04510, M\'exico D.F.\\[-1.6mm]
\small\texttt{delrio@leibniz.iimas.unam.mx}
\\[4mm]
\textbf{Mikhail Kudryavtsev}
\\
\small Department of Mathematics\\[-1.6mm]
\small Institute for Low Temperature Physics and Engineering\\[-1.6mm]
\small Lenin Av. 47, 61103\\[-1.6mm]
\small Kharkov, Ukraine\\[-1.6mm]
\small\texttt{kudryavtsev@onet.com.ua}
\\[4mm]
\textbf{Luis O. Silva}\thanks{%
Partially supported by CONACYT (M\'exico) through grant CB-2008-01-99100
}%
\\
\small Departamento de M\'{e}todos Matem\'{a}ticos y Num\'{e}ricos\\[-1.6mm]
\small Instituto de Investigaciones en Matem\'aticas Aplicadas y en Sistemas\\[-1.6mm]
\small Universidad Nacional Aut\'onoma de M\'exico\\[-1.6mm]
\small C.P. 04510, M\'exico D.F.\\[-1.6mm]
\small \texttt{silva@leibniz.iimas.unam.mx}
}
\date{}
\maketitle
\vspace{-4mm}
\begin{center}
\begin{minipage}{5in}
  \centerline{{\bf Abstract}} \bigskip We consider an inverse spectral
  problem for infinite linear mass-spring systems with different
  configurations obtained by changing the first mass. We give
  results on the reconstruction of the system from the spectra of two
  configurations.  Necessary and sufficient conditions
  for two real sequences to be the spectra of two modified
  systems are provided.
\end{minipage}
\end{center}
\thispagestyle{empty}
\end{titlepage}
\section{Introduction}
\label{sec:intro}
In this work we treat the two spectra inverse problem for Jacobi
operators in $l_2(\nats)$. The Jacobi operators considered here are
obtained from each other by a particular kind of rank-two
perturbation. The special form of the perturbation has a physical
motivation; it is the extension to the semi-infinite case of an
inverse problem for finite mass-spring systems studied in
\cite{delrio-kudryavtsev} and \cite{Ram}.

The Jacobi operator $J$ in the Hilbert space $l_2(\nats)$ is the operator
whose matrix representation with respect to the canonical basis in
$l_2(\nats)$ is a semi-infinite Jacobi matrix of the form
\begin{equation}
  \label{eq:jm-0}
  \begin{pmatrix}
    q_1 & b_1 & 0  &  0  &  \cdots
\\[1mm] b_1 & q_2 & b_2 & 0 & \cdots \\[1mm]  0  &  b_2  & q_3  &
b_3 &  \\
0 & 0 & b_3 & q_4 & \ddots\\ \vdots & \vdots &  & \ddots
& \ddots
  \end{pmatrix}\,,
\end{equation}
where $q_n\in\reals$ and $b_n>0$ for any $n\in\nats$ (see in
\cite{MR1255973} the definition of the matrix representation of an
unbounded symmetric operator). $J$ is closed by definition and it may
be self-adjoint or have deficiency indices (1,1). In this work we deal
with self-adjoint operators, so, if $J\ne J^*$, we consider its
self-adjoint extensions denoted $J^{(g)}$, where
$g\in\reals\cup\{\infty\}$ (see Definition~\ref{def:s-a-ext-def} a)).
If $J=J^*$ we assume $J^{(g)}=J$ for all $g\in\reals\cup\{\infty\}$
(see Definition~\ref{def:s-a-ext-def} b)).

The two spectra inverse problem for Jacobi operators $J^{(g)}$ takes
as input data the spectra of two operators in a operator family
obtained by perturbing $J^{(g)}$ in a certain way. The solution of the
problem is the finding of the matrix (\ref{eq:jm-0}) and the
``boundary condition at infinity'' $g$ if necessary. The case of the
operator family consisting of rank-one perturbations of a self-adjoint
Jacobi operator has been amply studied in
\cite{MR49:9676,MR499269,MR0221315} and, in the more general setting
of rank-one perturbations of $J^{(g)}$, in
\cite{weder-silva,MR1643529}.  Rank-one perturbations can be seen as a
change of the ``boundary condition at the origin'' for the
corresponding difference equation (see \cite[Appendix]{weder-silva}).
We remark that the case of finite Jacobi matrices has also been
thoroughly studied (see
\cite{Chu-Golub,deBoor-Golub,MR1616422,MR2102477,MR0382314}).

It is known that the dynamics of a finite mass-spring system is
characterized by the spectral properties of a finite Jacobi matrix
\cite{MR2102477}. Accordingly, in solving the inverse problem for
mass-spring systems mentioned above, \cite{Ram} provides necessary and
sufficient conditions for two point sets to be the spectra of two
finite Jacobi matrices corresponding to two mass-spring systems, one
of which has a mass and a spring modified.  The results of \cite{Ram}
are related to the study of microcantilevers
\cite{spletzer-et-al1,spletzer-et-al2}, which are modeled by a
spring-mass system whose masses and springs constants correspond to
the mechanical parameters of the system. The inverse problem treated
in \cite{Ram} could be used as a theoretical framework for the problem
of measuring micromasses with a help of microcantilevers
\cite{spletzer-et-al1,spletzer-et-al2}.

Let us consider a semi-infinite spring-mass system with masses
$\{m_j\}_{j=1}^\infty$ and spring constants $\{k_j\}_{j=1}^\infty$ as
in Fig.~\ref{fig:1}.
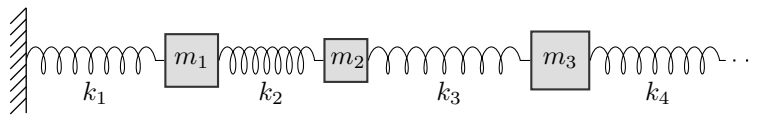
\begin{figure}[h]
\begin{center}
\begin{tikzpicture}
  [mass1/.style={rectangle,draw=black!80,fill=black!13,thick,inner sep=0pt,
   minimum size=7mm},
   mass2/.style={rectangle,draw=black!80,fill=black!13,thick,inner sep=0pt,
   minimum size=5.7mm},
   mass3/.style={rectangle,draw=black!80,fill=black!13,thick,inner sep=0pt,
   minimum size=7.7mm},
   wall/.style={postaction={draw,decorate,decoration={border,angle=-45,
   amplitude=0.3cm,segment length=1.5mm}}}]
  \node (mass3) at (7.1,1) [mass3] {\footnotesize$m_3$};
  \node (mass2) at (4.25,1) [mass2] {\footnotesize$\,m_2$};
  \node (mass1) at (2.2,1) [mass1] {\footnotesize$m_1$};
\draw[decorate,decoration={coil,aspect=0.4,segment
  length=2.1mm,amplitude=1.8mm}] (0,1) -- node[below=4pt]
{\footnotesize$k_1$} (mass1);
\draw[decorate,decoration={coil,aspect=0.4,segment
  length=1.5mm,amplitude=1.8mm}] (mass1) -- node[below=4pt]
{\footnotesize$k_2$} (mass2);
\draw[decorate,decoration={coil,aspect=0.4,segment
  length=2.5mm,amplitude=1.8mm}] (mass2) -- node[below=4pt]
{\footnotesize$k_3$} (mass3);
\draw[decorate,decoration={coil,aspect=0.4,segment
  length=2.1mm,amplitude=1.8mm}] (mass3) -- node[below=4pt]
{\footnotesize$k_4$} (9.3,1);
\draw[line width=.8pt,loosely dotted] (9.4,1) -- (9.8,1);
\draw[line width=.5pt,wall](0,1.7)--(0,0.3);
\end{tikzpicture}
\end{center}
\caption{Semi-infinite mass-spring system}\label{fig:1}
\end{figure}
By a standard reasoning (see
\cite{MR2102477,marchenko-new,mono-marchenko}) one verifies that the
infinite system of Fig.~\ref{fig:1} is modeled by the spectral
properties of the Jacobi operator $J$ with
\begin{equation}
\label{eq:spring-mass}
q_j = -\frac{k_{j+1}+k_j}{m_j}\,, \qquad
b_j=\frac{k_{j+1}}{\sqrt{m_j m_{j+1}}}\,,
\qquad j\in\nats\,.
\end{equation}
We remark that in \cite{MR2102477,marchenko-new,mono-marchenko} the obtained
matrix corresponds to $-J$. An alternative physical interpretation is
provided by a one dimensional harmonic crystal \cite[Sec. 1.5]{MR1711536}.

In this work we consider the spectrum of $J^{(g)}$ to be discrete (if
$J\ne J^*$ this is always the case). Below, in Remarks
\ref{rem:discrete-spectrum} and \ref{rem:nonselfadjoint} we comment on
matrices of the form (\ref{eq:jm-0}) whose corresponding operator
$J^{(g)}$ has discrete spectrum.

The discreteness of $\sigma(J^{(g)})$ implies that the movement of our
mechanical system is a superposition of harmonic oscillations whose
frequencies are the square roots of the modules of the eigenvalues.

Along with the self-adjoint operator $J^{(g)}$ we consider the family
of operators $J^{(g)}(\theta)$ ($\theta>0$) being self-adjoint
extensions of the Jacobi operator whose matrix representation with
respect to the canonical basis in $l_2(\nats)$ is
\begin{equation}
  \label{eq:jm-theta}
    \begin{pmatrix}
    \theta^2q_1 & \theta b_1 & 0  &  0  &  \cdots
\\[1mm] \theta b_1 & q_2 & b_2 & 0 & \cdots \\[1mm]  0  &  b_2  & q_3  &
b_3 &  \\
0 & 0 & b_3 & q_4 & \ddots\\ \vdots & \vdots &  & \ddots
& \ddots
  \end{pmatrix}\,.
\end{equation}
$J^{(g)}(\theta)$ ($\theta>0$) will be the family of perturbed Jacobi
operators. Note that the operators of the family are not obtained from
each other by a rank-one perturbation (see (\ref{eq:rank-two}) below).

Going from $J^{(g)}$ to $J^{(g)}(\theta)$ corresponds to changing the
first mass by $\Delta m=m_1(\theta^{-2}-1)$. In other words,
$\theta^2$ is the ratio of the original mass $m_1$ to the new mass
$m_1+\Delta m$. This is illustrated in Fig.~\ref{fig:2}. It is worth
mentioning that we also consider here the cases when $\Delta m<0$,
equivalently, $\theta>1$, although physical applications correspond to
$\theta<1$ \cite{spletzer-et-al1,spletzer-et-al2}.
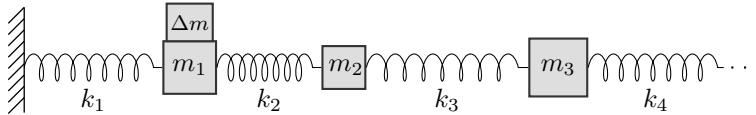
\begin{figure}[h]
\begin{center}
\begin{tikzpicture}
  [mass1/.style={rectangle,draw=black!80,fill=black!13,thick,inner sep=0pt,
   minimum size=7mm},
   mass2/.style={rectangle,draw=black!80,fill=black!13,thick,inner sep=0pt,
   minimum size=5.7mm},
   mass3/.style={rectangle,draw=black!80,fill=black!13,thick,inner sep=0pt,
   minimum size=7.7mm},
   dmass/.style={rectangle,draw=black!80,fill=black!13,thick,inner sep=0pt,
   minimum size=5mm},
   wall/.style={postaction={draw,decorate,decoration={border,angle=-45,
   amplitude=0.3cm,segment length=1.5mm}}}]
  \node (mass3) at (7.1,1) [mass3] {\footnotesize$m_3$};
  \node (mass2) at (4.25,1) [mass2] {\footnotesize$\,m_2$};
  \node (mass1) at (2.2,1) [mass1] {\footnotesize$m_1$};
  \node (dmass) at (2.2,1.6) [dmass] {\scriptsize$\,\Delta m\,$};
\draw[decorate,decoration={coil,aspect=0.4,segment
  length=2.1mm,amplitude=1.8mm}] (0,1) -- node[below=4pt]
{\footnotesize$k_1$}  (mass1);
\draw[decorate,decoration={coil,aspect=0.4,segment
  length=1.5mm,amplitude=1.8mm}] (mass1) -- node[below=4pt]
{\footnotesize$k_2$} (mass2);
\draw[decorate,decoration={coil,aspect=0.4,segment
  length=2.5mm,amplitude=1.8mm}] (mass2) -- node[below=4pt]
{\footnotesize$k_3$} (mass3);
\draw[decorate,decoration={coil,aspect=0.4,segment
  length=2.1mm,amplitude=1.8mm}] (mass3) -- node[below=4pt]
{\footnotesize$k_4$} (9.3,1);
\draw[line width=.8pt,loosely dotted] (9.4,1) -- (9.8,1);
\draw[line width=.5pt,wall](0,1.8)--(0,0.4);
\end{tikzpicture}
\end{center}
\caption{Perturbed semi-infinite mass-spring system}\label{fig:2}
\end{figure}

The problem of reconstructing the initial and the perturbed matrices
by their spectra can be then interpreted from the physical point of
view as the problem of finding the mechanical parameters of the
spring-mass system from the frequencies of its oscillations before
and after the modification.

We emphasize that, although the operators and the particular kind of
perturbation considered here were motivated by a physical system, the
general mathematical setting is consider throughout the work. Thus,
the entries in (\ref{eq:jm-0}) have no restriction other than $J$
being a Jacobi operator ($q_n\in\reals$, $b_n>0$) and $J^{(g)}$ having
discrete spectrum (see Remarks \ref{rem:discrete-spectrum},
\ref{rem:nonselfadjoint}). Note that $J$ is then not necessary
semibounded though it actually is when $J$ corresponds to a
mass-spring system.

This work is organized as follows. In Section~\ref{sec:preliminaries}
we lay down the notation, introduce the Jacobi operators and its
perturbations, and present some preparatory facts related with the
inverse spectral problems of such
operators. Section~\ref{sec:direct-problem} gives an account of the
spectral properties of the family of perturbed Jacobi operators
$J^{(g)}(\theta)$. The problem of reconstruction is treated in
Section~\ref{sec:reconstruction}. This section gives some necessary
conditions for the spectra of $J^{(g)}(\theta)$, provides an algorithm
for reconstruction of the matrix and establishes uniqueness of the
reconstruction.  Finally, Section~\ref{sec:nec-suf} gives necessary
and sufficient conditions for two sequences of real numbers to be the
spectra of $J^{(g)}$ and its perturbation $J^{(g)}(\theta)$
($\theta\ne 1$).

\section{Preliminaries}
\label{sec:preliminaries}
Let $\Upsilon$ be a second order symmetric difference
expression such that for any sequence $f=\{f_k\}_{k=1}^\infty$
\begin{align}
  \label{eq:initial-spectral}
 (\Upsilon f)_1&:= q_1 f_1 + b_1 f_2\,,\\
  \label{eq:recurrence-spectral}
  (\Upsilon f)_k&:= b_{k-1}f_{k-1} + q_k f_k + b_kf_{k+1}\,,
  \quad k \in \mathbb{N} \setminus \{1\},
\end{align}
where, for $n\in\mathbb{N}$, $b_n$ is positive and $q_n$ is real.  Let
$l_{\rm{fin}}(\mathbb{N})$ be the linear space of complex sequences with a
finite number of non-zero elements. In the Hilbert space
$l_2(\mathbb{N})$, let us consider the operator whose domain is
$l_{\rm{fin}}(\mathbb{N})$ and acts as the expression $\Upsilon$. This
operator is symmetric since it is densely defined and
Hermitian, and thus it is closable. Now, let $J$ be the closure of
this operator.

We have defined the operator $J$ so that the semi-infinite Jacobi
matrix (\ref{eq:jm-0}) is its matrix representation with respect to
the canonical basis $\{\delta_n\}_{n=1}^\infty$ in $l_2(\mathbb{N})$
(see \cite[Sec. 47]{MR1255973} for the definition of the matrix
representation of an unbounded symmetric operator). Indeed, $J$ is the
minimal closed symmetric operator satisfying
\begin{equation*}
\begin{array}{l}
\inner{\delta_n}{J\delta_n}=q_n\,,\quad
\inner{\delta_{n+1}}{J\delta_n}=\inner{\delta_n}{J\delta_{n+1}}=b_n\,,\\
\inner{J\delta_n}{\delta_{n+k}}
=\inner{\delta_{n}}{J\delta_{n+k}}=0\,,
\end{array}
\quad n\in\mathbb{N}\,,\,k\in\mathbb{N}\setminus\{1\}\,.
\end{equation*}
We shall refer to $J$ as the \emph{Jacobi operator} and to
(\ref{eq:jm-0}) as its associated matrix.

The operator $J^*$ turns out to be given by
\begin{equation*}
  \dom(J^*)=\{f\in l_2(\nats):\Upsilon f\in l_2(\nats)\},
  \qquad J^*f=\Upsilon f\,,
\end{equation*}
which follows directly from the definition of $J$ \cite[Chap.\,4
Sec.\,1.1]{MR0184042}, \cite[Thm.\,2.7]{MR1627806}.

If one gives the complex number $f_1$, the solution of
the difference equation,
\begin{equation*}
(\Upsilon f)= \zeta f\,,
\qquad \zeta \in \mathbb{C}\,,
\end{equation*}
is uniquely determined from (\ref{eq:initial-spectral}) and
(\ref{eq:recurrence-spectral}) by recurrence. For the elements of this
solution when $f_1=1$, the following notation is standard
\cite[Chap. 1, Sec. 2.1]{MR0184042}
\begin{equation*}
  P_{k-1}(\zeta):=f_k\,,\qquad k\in\mathbb{N}\,,
\end{equation*}
where the polynomial $P_k(\zeta)$ (of degree $k$) is referred
to as the $k$-th orthogonal polynomial of the first kind
associated with the matrix (\ref{eq:jm-0}). Now, let us
solve the difference equation
\begin{equation*}
  (\Upsilon f)_k= \zeta f_k, \qquad k\in \mathbb{N} \setminus \{1\}
\end{equation*}
under the assumption that $f_1=0$ and $f_2=b_1^{-1}$, and define
\begin{equation*}
  Q_{k-1}(\zeta):=f_k\,,\qquad k\in\mathbb{N}\,.
\end{equation*}
$Q_k(\zeta)$ is a polynomial of degree $k-1$ and it is called the $k$-th
orthogonal polynomial of the second kind associated with the matrix
(\ref{eq:jm-0}).

The sequence $P(\zeta):=\{P_{k-1}(\zeta)\}_{k=1}^\infty$ is not in
$l_{\rm{fin}}(\mathbb{N})$, but it may happen that
\begin{equation}
 \label{eq:generalized-eigenvector}
  \sum_{k=0}^\infty\abs{P_k(\zeta)}^2<\infty\,,
\end{equation}
in which case $P(\zeta)\in\ker(J^*-\zeta I)$. Since $J$ is symmetric,
if the series in (\ref{eq:generalized-eigenvector}) is convergent for
one $\zeta$ in the upper half plane $\complex_+$ (the lower half plane
$\complex_-$), then it is convergent in all $\complex_+$
($\complex_-$). Actually, because of the reality of the coefficients
of $P_{k-1}(\zeta)$ for all $k\in\nats$, the series in
(\ref{eq:generalized-eigenvector}) will then be convergent in all
$\complex\setminus\reals$ and $J$ has deficiency indices $(1,1)$. When
the series in (\ref{eq:generalized-eigenvector}) is divergent for one
$\zeta$ in $\complex\setminus\reals$, $J$ has deficiency indices
$(0,0)$ and the operator is self-adjoint since $J$ is closed. There
are known conditions on the matrix (\ref{eq:jm-0}) which guarantee
that $J$ is self-adjoint \cite[Addenda\,1]{MR0184042},
\cite[Chap.\,7,\,Thms.\,1.2--1.4]{MR0222718}.

We now introduce the operators that will be at the center of our
considerations in this work.
\begin{definition}
  \label{def:s-a-ext-def}
Let the operator $J^{(g)}$ be defined as follows:
\begin{enumerate}[a)]
\item In case $J\ne J^*$, define the sequence
$v(g)=\{v_k(g)\}_{k=1}^\infty$ such that $\forall k\in\mathbb{N}$
\begin{equation*}
  v_k(g):=P_{k-1}(0)+gQ_{k-1}(0)\,,
  \quad g\in\mathbb{R}
\end{equation*}
and
\begin{equation*}
  v_k(\infty):=Q_{k-1}(0)\,.
\end{equation*}
Let $J^{(g)}$ be the restriction of
$J^*$ to the set
\begin{equation*}
  \left\{f=\{f_k\}_{k\in\nats}\in\dom(J^*):
\lim_{k\to\infty}b_k(v_k(g)f_{k+1}-f_kv_{k+1}(g))=0\right\}\,.
\end{equation*}
When $g\in\reals\cup\{\infty\}$, $J^{(g)}$ runs over all self-adjoint
extensions of $J$. Moreover, different values of $g$ imply different
self-adjoint extensions \cite[Lemma 2.20]{MR1711536}.
\item In case $J=J^*$, let us define $J^{(g)}:=J$ for all
$g\in\reals\cup\{\infty\}$.
\end{enumerate}
\end{definition}
Alongside the operator $J^{(g)}$ we consider the operators $J_n^{(g)}$
($n\in\nats$) in the Hilbert space
$l_2(\nats)\ominus\Span\{\delta_1,\dots,\delta_n\}$ defined by
restricting $J^{(g)}$ to
$l_2(\nats)\ominus\Span\{\delta_1,\dots,\delta_n\}$. Thus, $J_n^{(g)}$
is a self-adjoint extension of the Jacobi operator whose associated
matrix is (\ref{eq:jm-0}) with the first $n$ columns and $n$ rows
removed.

Finally we introduce the perturbed operators $J^{(g)}(\theta)$. They
are defined as follows. Consider $J^{(g)}$ with fixed
$g\in\reals\cup\{\infty\}$ and take any $\theta>0$. Then
\begin{equation}
 \label{eq:rank-two}
  J^{(g)}(\theta):=J^{(g)} +
  q_1(\theta^2-1)\inner{\delta_1}{\cdot}\delta_1 +
  b_1(\theta-1)(\inner{\delta_1}{\cdot}\delta_2 +
  \inner{\delta_2}{\cdot}\delta_1)\,,
\end{equation}
where we take the inner product to be antilinear in its first
argument. By this definition $J^{(g)}(\theta)$ is a self-adjoint
extension of the Jacobi operator whose associated matrix is
(\ref{eq:jm-theta}). Note that $J^{(g)}(\theta)$ is a finite-rank
perturbation of  $J^{(g)}$ and thus $\dom(J^{(g)})=\dom(J^{(g)}(\theta))$.

Fix $g\in\reals\cup\{\infty\}$ and take the
resolution of the identity $E^{(g)}(t)$ of $J^{(g)}$, so
\begin{equation*}
  J^{(g)}=\int_\mathbb{R}tdE^{(g)}(t)\,.
\end{equation*}
Since $J^{(g)}$ is simple \cite[Sec.\,2.2\,Chap.\,4]{MR0184042}, it is
particularly useful to consider
the function
\begin{equation}
  \label{eq:spectral-measure}
  \rho^{(g)}(t):=\inner{\delta_1}{E^{(g)}(t)\delta_1}\,,\quad t\in\reals\,.
\end{equation}
It turns out that all the moments of the measure generated by $\rho^{(g)}$
are finite \cite[Thm.\,4.1.3]{MR0184042}, that is,
\begin{equation}
  \label{eq:moments}
  s_k=\int_\reals t^kd\rho^{(g)}(t)<\infty\qquad\forall k\in\nats\cup\{0\}\,,
\end{equation}
and the polynomials are
dense in $L_2(\reals,d\rho^{(g)})$ \cite[Thms.\,2.3.2,\,4.1.4]{MR0184042},
\cite[Prop.\,4.15]{MR1627806}.

In this work we also make use of the so-called Weyl $m$-function
\begin{equation}
  \label{eq:m-function}
  m^{(g)}(\zeta):=\inner{\delta_1}{(J^{(g)}-\zeta I)^{-1}
\delta_1}\,,\quad\zeta\not\in\sigma(J^{(g)})\,.
\end{equation}
The functions (\ref{eq:spectral-measure}) and (\ref{eq:m-function})
are related by the Borel transform, viz.,
\begin{equation*}
    m^{(g)}(\zeta) =\int_{\mathbb{R}}\frac{d\rho^{(g)}(t)}{t-\zeta}\,,
\end{equation*}
so $m^{(g)}$ is a Herglotz function, i.\,e.,
\begin{equation*}
  \frac{\im m^{(g)}(\zeta)}{\im \zeta}>0\,,\qquad\im\zeta>0\,.
\end{equation*}
Using the von Neumann expansion for the resolvent
(cf.\cite[Chap.\,6,\,Sec.\,6.1]{MR1711536})
\begin{equation*}
  (J^{(g)}-\zeta I)^{-1}=
  -\sum_{k=0}^{N-1}\frac{(J^{(g)})^k}{\zeta^{k+1}}
  +\frac{(J^{(g)})^N}{\zeta^{N}}
  (J^{(g)}-\zeta I)^{-1}\,,
\end{equation*}
where $\zeta\in\mathbb{C}\setminus\sigma(J^{(g)})$,
one can easily obtain the following asymptotic formula
\begin{equation}
  \label{eq:m-asympt}
  m^{(g)}(\zeta)=-\frac{1}{\zeta}-\frac{q_1}{\zeta^2}
  -\frac{b_1^2+q_1^2}{\zeta^3}
  +O(\zeta^{-4})\,,
\end{equation}
as $\zeta\to\infty$ ($\im \zeta\ge \epsilon$, $\epsilon>0$).

The inverse Stieltjes transform allows to recover the spectral
function (\ref{eq:spectral-measure}) from its corresponding Weyl
$m$-function (\ref{eq:m-function}). So they are in one-to-one
correspondence. Furthermore, either (\ref{eq:spectral-measure}) or
(\ref{eq:m-function}) uniquely determines the Jacobi operator
$J^{(g)}$, i.\,e., the matrix (\ref{eq:jm-0}) and the parameter $g$ in
the non-self-adjoint case.  Indeed, there are two general methods for
recovering the matrix (\ref{eq:jm-0}) that work without any assumption
on the spectrum. One method, developed in \cite{MR1616422} (see also
\cite{MR1643529}), makes use of the asymptotic behavior of the Weyl
$m$-function and the Riccati equation \cite[Eq.\,2.15]{MR1616422},
\cite[Eq.\,2.23]{MR1643529},
\begin{equation}
  \label{eq:riccati}
    b_n^2 m_n^{(g)}(\zeta)=
    q_n-\zeta-\frac{1}{m_{n-1}^{(g)}(\zeta)}\,,\quad n\in\mathbb{N}\,,
\end{equation}
where $m_n^{(g)}(\zeta)$ is the Weyl $m$-function of the
Jacobi operator $J_n^{(g)}$ ($m_0=m$).

The other method of reconstruction (see
\cite[Chap.\,7,\,Sec.\,1.5]{MR0222718} and, particularly,
\cite[Chap.\,7,\,Thm.\,1.11]{MR0222718}) has its starting point in the
sequence $\{t^k\}_{k=0}^\infty$, $t\in\mathbb{R}$. From
(\ref{eq:moments}) all the elements of the sequence
$\{t^k\}_{k=0}^\infty$ are in $L_2(\mathbb{R},d\rho^{(g)})$ and one
can apply, in this Hilbert space, the Gram-Schmidt procedure of
orthonormalization to the sequence $\{t^k\}_{k=0}^\infty$.  One, thus,
obtains a sequence of polynomials $\{P_k(t)\}_{k=0}^\infty$ normalized
and orthogonal in $L_2(\mathbb{R},d\rho^{(g)})$. These polynomials
satisfy a three term recurrence equation
\cite[Chap.\,7,\,Sec.\,1.5]{MR0222718}, \cite[Sec.\,1]{MR1627806}
\begin{align}
\label{eq:favard-system1}
      tP_{k-1}(t) &= b_{k-1}P_{k-2}(t) +  q_k
      P_{k-1}(t) +  b_k P_k(t)\,,
\quad k \in \mathbb{N} \setminus \{1\}\,,\\
\label{eq:favard-system2}
 tP_0(t) &=  q_1 P_0(t) +  b_1 P_1(t)\,,
\end{align}
where all the coefficients $b_k$ ($k\in\mathbb{N}$) turn out to be
positive and $q_k$ ($k\in\mathbb{N}$) are real numbers.  The system
(\ref{eq:favard-system1}) and (\ref{eq:favard-system2}) defines a
Jacobi matrix which is the matrix representation of either
$J^{(g)}$ or a restriction of $J^{(g)}$ depending on whether $J=
J^*$ or not.

The function (\ref{eq:m-function}), equivalently
(\ref{eq:spectral-measure}), determines the parameter $g$
which defines the self-adjoint extension when the
reconstructed matrix turns out to be the matrix representation of a
non-self-adjoint operator. Indeed, consider a pole $\gamma$ of
$m^{(g)}$ (there is always one when $J\ne J^*$) and evaluate
$P_k(\gamma)$, $k\in\nats$. Then either
\begin{equation*}
  \lim_{k\to\infty}b_k(Q_{k-1}(0)P_{k}(\gamma)-P_{k-1}(\gamma)Q_{k-1}(0))=0,
\end{equation*}
which means that $g=\infty$, or
\begin{equation*}
  g=\frac{\lim_{k\to\infty}b_k(P_{k-1}(0)P_{k}(\gamma)-P_{k-1}(\gamma)P_{k-1}(0))}
  {\lim_{k\to\infty}b_k(Q_{k-1}(0)P_{k}(\gamma)-P_{k-1}(\gamma)Q_{k-1}(0))}\,.
\end{equation*}
The details of this recipe are explained for instance in
\cite[Sec.\,2]{weder-silva}.

Since any simple self-adjoint operator in an infinite dimensional
Hilbert space is unitarily equivalent to some operator $J=J^*$
\cite[Thm.\,4.2.3]{MR0184042}, \cite[Sec.\,69]{MR1255973}, in the case
$J=J^*$, $\sigma(J^{(g)})$ may be any non-empty closed infinite set in
$\reals$. In particular $J^{(g)}$ may have discrete spectrum, that is,
$\sigma_{ess}(J^{(g)})=\emptyset$. When $J\ne J^*$, this is always the
case, that is all self-adjoint extensions $J^{(g)}$ of the
non-self-adjoint operator $J$ have discrete spectrum
\cite[Lem.\,2.19]{MR1711536}.

Assume that $J$ has discrete spectrum (this always happen if $J\ne
J^*$), so the spectrum is a sequence of real numbers,
$\{\lambda_k\}_k$, without finite points of accumulation.
The simplicity of $J^{(g)}$ implies
that all eigenvalues are of multiplicity one. In this case the
function $\rho^{(g)}(t)$, defined by (\ref{eq:spectral-measure}), can be
written as follows
\begin{equation}
  \label{eq:rho-discrete}
   \rho^{(g)}(t)=\sum_{\lambda_k< t}\frac{1}{\alpha_k}\,,
\end{equation}
where the coefficients $\{\alpha_k\}_k$ are called the normalizing
constants and according to \cite[Chap.\,7,\,Thm.\,1.17]{MR0222718} are
given by
\begin{equation}
  \label{eq:def-normalizing}
  \alpha_n=\sum_{k=0}^\infty\abs{P_k(\lambda_n)}^2\,.
\end{equation}
Thus, from (\ref{eq:rho-discrete}) and (\ref{eq:m-function}) one has
that
\begin{equation}
  \label{eq:m-discrete}
  m^{(g)}(\zeta)=\sum_{k}\frac{1}{\alpha_k(\lambda_k-\zeta)}\,.
\end{equation}
\begin{remark}
  \label{rem:m-zeros-poles}
  In the case of discrete spectrum, the set of poles of the
  meromorphic Weyl $m$-function coincides with $\sigma(J^{(g)})$. By
  (\ref{eq:riccati}), the set of zeros coincides with
  $\sigma(J_1^{(g)})$. The zeros and poles of the Weyl $m$-function
  are simple and interlace as occurred to any meromorphic Herglotz
  function. Interlacing means that between two contiguous poles there
  is exactly one zero and between two contiguous zeros there is
  exactly one pole (see the proof of \cite[Chap.\,7,\,Thm.\,1]{MR589888}).
\end{remark}
\begin{remark}
  \label{rem:discrete-simple}
  By elementary perturbation theory (Weyl theorem), $J^{(g)}$ has
  discrete spectrum if and only if $J^{(g)}(\theta)$ has discrete
  spectrum. Note that $J^{(g)}(\theta)$ has simple spectrum since it
  is a self-adjoint extension of a Jacobi operator.
\end{remark}
\begin{remark}
  \label{rem:discrete-spectrum}
  Let us comment briefly on the criteria for discreteness of
  $\sigma(J^{(g)})$ on the basis of the matrix entries in
  (\ref{eq:jm-0}) when $J=J^*$.  Consider a matrix whose main diagonal
  is a sequence $\{q_k\}_{k=1}^\infty$ of pairwise distinct real
  numbers without finite accumulation points and the sequence defining
  the off-diagonals $\{b_k\}_{k=1}^\infty$ is such that $b_k=o(q_k)$
  as $k\to\infty$. Then, it can be shown that $J$ is the sum of the
  operator $D$ whose matrix representation is
  $\diag\{q_k\}_{k=1}^\infty$ and a perturbation relatively compact
  with respect to $D$. By perturbation theory, $J$ is thus self-adjoint
  and has discrete spectrum. Of course there are other examples of
  self-adjoint Jacobi operators having discrete spectrum and whose
  matrix representation diagonals do not satisfy the conditions just
  given (see for instance \cite{naboko,toloza-growing-weights})
\end{remark}
\begin{remark}
  \label{rem:nonselfadjoint}
  There are conditions on the entries of (\ref{eq:jm-0}) which
  guarantee that $J\ne J^*$ (see for instance
  \cite[Addenda\,1]{MR0184042} and
  \cite[Thm.\,7.1.5]{MR0222718}). Thus, for (\ref{eq:jm-0}) satisfying
  those conditions, $J^{(g)}$ has discrete spectrum
  \cite[Lem.\,2.19]{MR1711536}.
\end{remark}
\begin{remark}
  \label{rem:discrete-mass-spring}
  Consider the mass-spring system of the Introduction. On the basis of
  Remarks \ref{rem:discrete-spectrum}, \ref{rem:nonselfadjoint}, and
  by means of the recurrence equations given below in Remark
  \ref{rem:inverse-mass-spring}, one could construct a mass-spring
  system whose corresponding operator $J^{(g)}$ has discrete spectrum.
\end{remark}

\section{Direct spectral analysis of $J^{(g)}$ and $J^{(g)}(\theta)$}
\label{sec:direct-problem}
We begin this section by noting that
\begin{equation*}
  J_1^{(g)}=J_1^{(g)}(\theta)\,, \qquad \forall\theta>0\,.
\end{equation*}
Fix $g\in\reals\cup\{\infty\}$ and consider the Weyl $m$-functions
$m^{(g)}$, $m^{(g,\theta)}$ of the operators $J^{(g)}$ and
$J^{(g)}(\theta)$. Therefore, taking into account that $m_1^{(g)}$ and
$m_1^{(g,\theta)}$ coincide, (\ref{eq:riccati}) implies that
\begin{equation}
  \label{eq:aux-m-m-theta}
  \theta^2\left(\zeta+\frac{1}{m^{(g)}(\zeta)}\right)=
\zeta+\frac{1}{m^{(g,\theta)}(\zeta)}\,,
\end{equation}
Let us now
consider the function
\begin{equation}
  \label{eq:m-goth-def}
  \mathfrak{m}(\zeta):=\frac{m^{(g)}(\zeta)}{m^{(g,\theta)}(\zeta)}\,.
\end{equation}
\begin{remark}
  \label{rem:zeros-poles}
  In view of Remark~\ref{rem:discrete-simple}, if $J^{(g)}$ has
  discrete spectrum, the function $\mathfrak{m}$ is meromorphic by
  (\ref{eq:m-goth-def}). Since the zeros of $m^{(g)}$ and
  $m^{(g,\theta)}$ are the same (see
  Remark~\ref{rem:m-zeros-poles}), it follows that for all $\theta>0$
  the set of poles of $\mathfrak{m}$ is a subset of $\sigma(J^{(g)})$,
  while $\sigma(J^{(g)}(\theta))$ contains all the zeros of
  $\mathfrak{m}$. Observe also that, from (\ref{eq:aux-m-m-theta}),
  $0\in\sigma(J^{(g)})$ if and only if
  $0\in\sigma(J^{(g)}(\theta))$. Moreover, whenever $\theta\ne 1$,
  (\ref{eq:aux-m-m-theta}) implies that the sets $\sigma(J^{(g)})$ and
  $\sigma(J^{(g)}(\theta))$ can intersect only at $0$.
\end{remark}
\begin{remark}
  \label{rem:continuous}
  By \cite[Chap.\,7,\,Thm.\,3.9]{MR0407617} the zeros
  of $\mathfrak{m}$ are analytic functions of the parameter
  $\theta$. The same is true for the eigenvectors of $J^{(g)}(\theta)$.
\end{remark}
\begin{proposition}
  \label{prop:derivative}
  Let $J^{(g)}$ have discrete spectrum and let $\{\lambda_k(\theta)\}_k$
  be the set of eigenvalues of $J^{(g)}(\theta)$ ($\theta>0$). For a fixed
  $k$ the following holds
\begin{equation*}
  \frac{d}{d\theta}\lambda_k(\theta)=
\frac{2\lambda_k(\theta)}{\theta\alpha_k(\theta)}\,,
\end{equation*}
where $\alpha_k(\theta)$ is the normalizing constant corresponding
to $\lambda_k(\theta)$.
\end{proposition}
\begin{proof}
  Let us denote by $f(\theta)$ the eigenvector of $J^{(g)}(\theta)$
  corresponding to $\lambda_k(\theta)$. We assume that $f(\theta)$
  is normalized in such a way that
  \begin{equation}
    \label{eq:eigenvalue-1}
    \inner{\delta_1}{f(\theta)}=1\,.
  \end{equation}
  Pick any small real $\tau$ (it suffices that
  $\abs{\tau}<\theta$). Then, taking into account that $\dom
  (J^{(g)})=\dom (J^{(g)}(\theta))$ and the self-adjointness of
  $J^{(g)}(\theta)$ for any $\theta>0$, we have that
  \begin{align*}
    (\lambda_k(\theta+\tau)-\lambda_k(\theta))
    \inner{f(\theta)}{f(\theta+\tau)}&=
    \inner{f(\theta)}{J^{(g)}(\theta+\tau)f(\theta+\tau)}\\
    &- \inner{J^{(g)}(\theta)f(\theta)}{f(\theta+\tau)}\\
    &=\inner{f(\theta)}
  {(J^{(g)}(\theta+\tau)
    -J^{(g)}(\theta)+J^{(g)}(\theta))f(\theta+\tau)}\\
  &-\inner{J^{(g)}(\theta)f(\theta)}{f(\theta+\tau)}\\
  &=\inner{f(\theta)}
  {(J^{(g)}(\theta+\tau)-J^{(g)}(\theta))f(\theta+\tau)}\,.
  \end{align*}
From (\ref{eq:eigenvalue-1}) it follows that the
entries $f(\theta+\tau)$ and $f(\theta)$ are the
polynomials of the first kind associated to the matrix of
$J^{(g)}(\theta+\tau)$ and $J^{(g)}(\theta)$, so
\begin{equation*}
  f_2(\theta+\tau)=\frac{\lambda_k(\theta+\tau)-(\theta+\tau)^2q_1}
  {(\theta+\tau)b_1}\,,\qquad
  f_2(\theta)=\frac{\lambda_k(\theta)-\theta^2q_1}{\theta b_1}\,,
\end{equation*}
Now,  taking into account these last equalities and
(\ref{eq:eigenvalue-1}), together with
 \begin{equation*}
   J^{(g)}(\theta+\tau)-J^{(g)}(\theta)=
    \begin{pmatrix}
    (2\theta\tau+\tau^2)q_1 & \tau b_1 & 0  &  0  &  \cdots
\\[1mm] \tau b_1 & 0 & 0 & 0 & \cdots \\[1mm]  0  &  0  & 0  &
0 &  \\
0 & 0 & 0 & 0 & \ddots\\ \vdots & \vdots &  & \ddots
& \ddots
  \end{pmatrix}\,,
 \end{equation*}
one obtains that
\begin{equation*}
  (\lambda_k(\theta+\tau)-\lambda_k(\theta))
    \inner{f(\theta)}{f(\theta+\tau)}=
    \tau\left(\frac{\lambda_k(\theta+\tau)}{\theta+\tau}+
\frac{\lambda_k(\theta)}{\theta}\right)
\end{equation*}
Therefore, on the basis of Remark~\ref{rem:continuous}, one has
\begin{equation*}
  \lim_{\tau\to
  0}\frac{\lambda_k(\theta+\tau)-\lambda_k(\theta)}{\tau}=
 \lim_{\tau\to
  0}\frac{1}{\inner{f(\theta)}{f(\theta+\tau)}}
\left(\frac{\lambda_k(\theta+\tau)}{\theta+\tau}+
\frac{\lambda_k(\theta)}{\theta}\right)=
  \frac{2\lambda_k(\theta)}{\theta\alpha_k(\theta)}\,.
\end{equation*}
\end{proof}
The proposition below can be proven by means of
Remark~\ref{rem:zeros-poles}, \ref{rem:continuous}, and
Proposition~\ref{prop:derivative}. However, we present an alternative
proof based on the following expression
\begin{equation}
  \label{eq:m-through-m}
  \mathfrak{m}(\zeta)=\zeta(\theta^2-1)m^{(g)}(\zeta)+\theta^2\,,
\end{equation}
which follows from (\ref{eq:aux-m-m-theta}) and (\ref{eq:m-goth-def}).
\begin{proposition}
 \label{prop:interlacing}
 Fix $g\in\reals\cup\{\infty\}$ and let $J^{(g)}$ have discrete
 spectrum.  The spectra $\sigma(J^{(g)})$, $\sigma(J^{(g)}(\theta))$
 interlace in $\mathbb{R}_+$ and $\mathbb{R}_-$. Moreover,
 $\sigma(J^{(g)}(\theta))$ in $\reals_+$ ($\reals_-$) is shifted with
 respect to $\sigma(J^{(g)})$ to the left (right) if $\theta<1$, and
 to the right (left) if $\theta>1$.
\end{proposition}
\begin{proof}
  In view of Remark~\ref{rem:zeros-poles}, one only needs to verify
  that between two positive and contiguous eigenvalues of $J^{(g)}$
  there is only one eigenvalue of $J^{(g)}(\theta)$ and viceversa.
  Take two positive and contiguous eigenvalues of $\sigma(J^{(g)})$,
  $\lambda<\widetilde{\lambda}$.  Due to (\ref{eq:m-discrete}), one
  has
  \begin{equation}
    \label{eq:m-asympt-in-eigenvalue}
    \lim_{\substack{t\to\widetilde{\lambda}^- \\
        t\in\reals}}m^{(g)}(t)=+\infty\,,
\qquad
    \lim_{\substack{t\to\lambda^+ \\ t\in\reals}}m^{(g)}(t)=-\infty\,.
  \end{equation}
  Now, in (\ref{eq:m-through-m}) assume that $\theta>1$.  Thus,
  because of the positivity of $\lambda,\widetilde{\lambda}$,
  (\ref{eq:m-through-m}) and (\ref{eq:m-asympt-in-eigenvalue}) imply that
  \begin{equation*}
     \lim_{\substack{t\to\widetilde{\lambda}^- \\
         t\in\reals}}\mathfrak{m}(t)
    =+\infty\,,\qquad
    \lim_{\substack{t\to\lambda^+ \\ t\in\reals}}\mathfrak{m}(t)=-\infty\,.
  \end{equation*}
  Since $\mathfrak{m}$ is analytic on
  the interval $(\lambda,\widetilde{\lambda})$, it should cross the
  0-axis an odd number of times. If it crosses this axis three or more
  times as in Fig. \ref{fig:3} ($a$), then, by
  Remarks~\ref{rem:m-zeros-poles} and \ref{rem:zeros-poles}, there are
  at least two elements of $\sigma(J_1^{(g)})$ in
  $(\lambda,\widetilde{\lambda})$. But, because of
  Remark~\ref{rem:m-zeros-poles}, this would contradict the fact that
  $\lambda,\widetilde{\lambda}$ are contiguous.
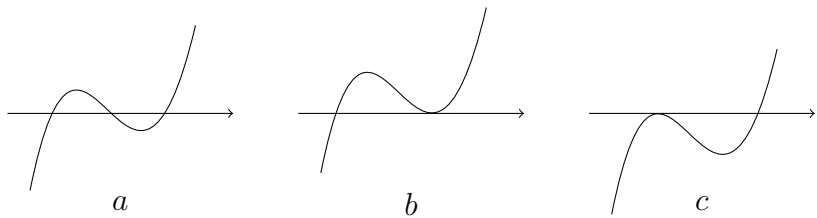
\begin{figure}[h]
\begin{center}
  \begin{tikzpicture}
    \draw[yshift=-15] (-.7,-.5) .. controls (0,3) and (.65,-1.9)
    .. (1.5,1.7);
    \draw[->] (-1,0) -- (2,0);
    \path (0.5,-1.2) node {$a$};
    \draw[xshift=110,yshift=-8.3] (-.7,-.5) .. controls (0,3) and (.65,-1.9)
    .. (1.5,1.7);
    \draw[->,xshift=110] (-1,0) -- (2,0);
    \path[xshift=110] (0.5,-1.2) node {$b$};
    \draw[xshift=220,yshift=-24] (-.7,-.5) .. controls (0,3) and (.65,-1.9)
    .. (1.5,1.7);
    \draw[->,xshift=220] (-1,0) -- (2,0);
    \path[xshift=220] (0.5,-1.2) node {$c$};
  \end{tikzpicture}
\end{center}
\caption{Impossible crossings of the 0-axis by $\mathfrak{m}$}\label{fig:3}
\end{figure}
Observe that one should discard the possibility of one crossing of the
0-axis and a tangential touch of it as in Fig. \ref{fig:3} ($b$) and
($c$). But again the impossibility of this follows from the fact that
the poles of $m^{(g,\theta)}$ are simple (see
Remark~\ref{rem:m-zeros-poles}). Analogously, between two contiguous
eigenvalues of $J^{(g)}(\theta)$, the function
$\frac{1}{\mathfrak{m}}\eval{\reals}$ crosses the $0$-axis exactly
once.  Thus, the interlacing in $\reals_+$ has been established. By the
same token, the spectra interlace in $\reals_-$. The case $\theta <1$
is treated in a similar way. The second assertion follows directly
from Proposition~\ref{prop:derivative}.
\end{proof}
\begin{remark}
  We note that $\sigma(J^{(g)})\cap\reals_+$,
  $\sigma(J^{(g)})\cap\reals_-$, may be finite or empty.
\end{remark}
\section{Inverse spectral analysis for $J^{(g)}$ and $J^{(g)}(\theta)$}
\label{sec:reconstruction}
In this section we find some necessary conditions for the spectra of
$J^{(g)}(\theta)$ ($\theta>0$). Also we provide a reconstruction algorithm
of the Jacobi matrix and establish uniqueness of the
reconstruction. Some of the formulae obtained in this section (see for
instance Corollary~\ref{cor:theta}) have an
analogous one in the finite case \cite{delrio-kudryavtsev}, \cite{Ram}.

A central part of our approach is the Weyl
$m$-function and its properties. We begin our discussion by setting
out a convention for enumerating the elements of the spectra.
\begin{convention}
\label{con:enumeration}
For a given countable set of real numbers $S$ without finite points of
accumulation, let $M$ be an infinite subset of consecutive integers
such that there is a one-to-one increasing function $h:M\to S$ with
the property that, $h^{-1}(0)=\{0\}$ when $0$ is in $S$.
Thus, $M$
is semi-bounded from above (below) if and only if
the same holds for $S$. We write $S=\{\lambda_k\}_{k\in M}$, where
$\lambda_k=h(k)$. Note that in the sequence $\{\lambda_k\}_{k\in M}$
only  $\lambda_0$ is allowed to be zero. Thus, if $-1,1\in M$,
then
\begin{equation*}
  \lambda_{-1} < 0 < \lambda_1\,.
\end{equation*}
In the sequel, the spectra of all operators will be enumerated
according to this convention.

When $\{\lambda_k\}_{k\in M}$ is considered together with a sequence
interlacing with it, we use the same set $M$ for enumerating both
sequences. For instance, if $\{\lambda_k\}_{k\in M}$ and
$\{\mu_k\}_{k\in M}$ are interlacing and not semi-bounded, then one can
assume that
\begin{equation*}
  \lambda_k<\mu_k<\lambda_{k+1}\,,\quad\forall k\in M.
\end{equation*}
\end{convention}
The following auxiliary result can be found in
\cite[Sec.\,4]{weder-silva}. We sketch the proof here for the reader's
convenience.
\begin{lemma}
  Let $J^{(g)}$ have discrete spectrum and assume that
  $\sigma(J^{(g)})=\{\lambda_k\}_{k\in M}$, and
  $\sigma(J_1^{(g)})=\{\eta_k\}_{k\in M}$. Then, the following formula
  holds for the Weyl $m$-function of $J^{(g)}$
  \begin{equation}
 \label{eq:levin-herglotz-gen}
    m^{(g)}(\zeta)=C \frac{\zeta-\eta_0}{\zeta-\lambda_0}
  \prod_{\substack{k\in M\\k\ne 0}} \left(1-\frac{\zeta}{\eta_k}\right)
  \left(1-\frac{\zeta}{\lambda_k}\right)^{-1}\,,
  \end{equation}
Moreover, $C<0$ and
\begin{equation}
  \label{eq:enum-zeros-poles-alt}
  \eta_k<\lambda_k<\eta_{k+1}\quad\forall k\in M\,,
\end{equation}
if $\sigma(J^{(g)})$ is semi-bounded from above,
while, $C>0$ and
\begin{equation}
  \label{eq:enum-zeros-poles}
  \lambda_k<\eta_k<\lambda_{k+1}\quad\forall k\in M
\end{equation}
otherwise.
\end{lemma}
\begin{proof}
  Assume first that $\sigma(J^{(g)})$ is semi-bounded from
  below. Since the greatest lower bound of $J$ does not exceed the
  greatest lower bound of $J_1^{(g)}$, the smallest element of
  $\{\lambda_k\}_{k\in M}$ is less than the smallest of
  $\{\eta_k\}_{k\in M}$ (see
  \cite[Chap.\,6,\,Sec.\,1.3]{MR1192782}). Thus one can enumerate the
  sequences $\{\lambda_k\}_{k\in M}$ and $\{\eta_k\}_{k\in M}$ so that
  they obey our convention and (\ref{eq:enum-zeros-poles}). According
  to \cite[Chap.\,7,\,Thm.\,1]{MR589888},
  (\ref{eq:levin-herglotz-gen}) holds with $C>0$.

Clearly, when $\sigma(J^{(g)})$ is not semi-bounded, the sequences can be
arranged to obey (\ref{eq:enum-zeros-poles}), and then
(\ref{eq:levin-herglotz-gen}) holds with $C>0$.

Now suppose that $\sigma(J^{(g)})$ is semi-bounded from above. Then
$\sigma(-J^{(g)})$ is semi-bounded from below and, consequently, the
greatest of $\{\eta_k\}_{k\in M}$ is less than the greatest of
$\{\lambda_k\}_{k\in M}$. Thus $\{\lambda_k\}_{k\in M}$, and
$\{\eta_k\}_{k\in M}$ cannot be arranged according to
(\ref{eq:enum-zeros-poles}). However, we are still able to use
(\ref{eq:enum-zeros-poles}) for arranging the zeros and poles
of the meromorphic Herglotz function $-\frac{1}{m^{(g)}}$, that is, we use
(\ref{eq:enum-zeros-poles-alt}).
Therefore \cite[Chap.\,7,\,Thm.\,1]{MR589888} gives
\begin{equation*}
  -\frac{1}{m^{(g)}(\zeta)}= \widetilde{C}
\frac{\zeta-\lambda_0}{\zeta-\eta_0}
  \prod_{\substack{k\in M\\k\ne 0}} \left(1-\frac{\zeta}{\lambda_k}\right)
  \left(1-\frac{\zeta}{\eta_k}\right)^{-1}\,,\qquad \widetilde{C}>0\,,
\end{equation*}
For completing the proof it only remains to note that the last
equation can be rewritten as asserted in the lemma. The infinite
product in (\ref{eq:levin-herglotz-gen}) is convergent because of
(\ref{eq:enum-zeros-poles-alt}) (see the proof of
\cite[Chap.\,7,\,Thm.\,1]{MR589888}).
\end{proof}
Another auxiliary simple result to be use later is the following lemma.
\begin{lemma}
  \label{lem:uniform-convergence}
  Let $J^{(g)}$ have discrete spectrum and $\{\lambda_k(\theta)\}_k$
  be the set of eigenvalues of $J^{(g)}(\theta)$. Then,
  the series
  \begin{equation}
    \label{eq:s-1-moment}
    \sum_{k\in M}\frac{\lambda_k(\theta)}{\alpha_k(\theta)}
  \end{equation}
  converges uniformly in $[\theta_1,\theta_2]\subset\reals_+$ to
  $s_1(\theta)$ (see (\ref{eq:moments})).
\end{lemma}
\begin{proof}
  From (\ref{eq:moments}) and (\ref{eq:rho-discrete}), it follows
  that the series converges pointwise to $s_1(\theta)$. The series
  \begin{equation}
    \label{eq:s-2-moment}
    \sum_{k\in M}\frac{\lambda_k^2(\theta)}{\alpha_k(\theta)}
  \end{equation}
  converges also pointwise to the function $s_2(\theta)$. Since this
  function is continuous in $[\theta_1,\theta_2]$, then
  (\ref{eq:s-2-moment}) is uniformly convergent in that interval (see
  \cite[Sec.\,1.31]{titchmarsh_functions}). Now, for any
  $\theta\in[\theta_1,\theta_2]$ and $\abs{\lambda_k}>1$, one has
  \begin{equation*}
    \abs{\lambda_k}<\lambda_k^2\,,
  \end{equation*}
  so (\ref{eq:s-1-moment}) is uniformly convergent in
  $[\theta_1,\theta_2]$.
\end{proof}
\begin{remark}
  \label{rem:arrange_sequences}
  Proposition~\ref{prop:interlacing} tells that the interlacing of the
  sequences $\sigma(J^{(g)})=\{\lambda_k\}_k$ and
  $\sigma(J(\theta))=\{\mu_k\}_k$ is different in $\reals_+$ and
  $\reals_-$. So let us agree to enumerate the sequences according to
  our convention (the subscripts of the sequences run over $M$ and
  only the eigenvalues with subscript equal zero are allowed to be
  zero) and obeying
  \begin{equation*}
    \lambda_k<\mu_k<\lambda_{k+1} \quad\text{in}\ \reals_+\,,
    \qquad
    \mu_k<\lambda_k<\mu_{k+1} \quad\text{in}\ \reals_-\,,
  \end{equation*}
when $\theta>1$, and
\begin{equation*}
  \mu_k<\lambda_k<\mu_{k+1} \quad\text{in}\ \reals_+\,,
  \qquad
  \lambda_k<\mu_k<\lambda_{k+1} \quad\text{in}\ \reals_-\,,
\end{equation*}
if $\theta<1$.
\end{remark}
\begin{proposition}
  \label{prop:convergence-eigenvalues}
  Fix $g\in\reals\cup\{\infty\}$ and $0<\theta_1<\theta_2$. Let
  $J^{(g)}$ have discrete spectrum and assume that
  $\sigma(J^{(g)}(\theta_1))=\{\lambda_k\}_{k\in M}$ and
  $\sigma(J^{(g)}(\theta_2))=\{\mu_k\}_{k\in M}$, where the sequences
  have been arranged according to
  Remark~\ref{rem:arrange_sequences}. Then,
  \begin{equation*}
    \sum_{k\in M}(\mu_k-\lambda_k)=q_1(\theta_2^2-\theta_1^2)\,.
  \end{equation*}
\end{proposition}
\begin{proof}
  Observe that from Proposition~\ref{prop:derivative} it follows that
  \begin{equation*}
    \mu_k-\lambda_k=2\int_{\theta_1}^{\theta_2}
    \frac{\lambda_k(\theta)d\theta}{\theta\alpha_k(\theta)}\,.
  \end{equation*}
Consider a sequence $\{M_n\}_{n=1}^\infty$ of subsets of $M$, such
  that $M_n\subset M_{n+1}$ and $\cup_nM_n=M$.
Thus
\begin{equation*}
  \sum_{k\in M}(\mu_k-\lambda_k)=2
\lim_{n\to\infty}\int_{\theta_1}^{\theta_2}\left(
\sum_{k\in M_n}\frac{\lambda_k(\theta)}{\alpha_k(\theta)}
\right)\frac{d\theta}{\theta}
\end{equation*}
By Lemma \ref{lem:uniform-convergence} and the fact that
\begin{equation*}
  s_1(\theta)=\inner{\delta_1}{J^{(g)}(\theta)\delta_1}=q_1\theta^2\,,
\end{equation*}
one obtains
\begin{equation*}
  \sum_{k\in M}(\mu_k-\lambda_k)=
2q_1\int_{\theta_1}^{\theta_2}\theta d\theta
  =q_1(\theta_2^2-\theta_1^2)
\end{equation*}
\end{proof}
\begin{proposition}
  \label{prop:m-goth-actual-form}
  Fix $g\in\reals\cup\{\infty\}$ and $0<\theta\ne 1$. Let $J^{(g)}$ have
  discrete spectrum and assume that
  $\sigma(J^{(g)})=\{\lambda_k\}_{k\in M}$ and
  $\sigma(J^{(g)}(\theta))=\{\mu_k\}_{k\in M}$, where the sequences
  have been arranged according to
  Remark~\ref{rem:arrange_sequences}. Then,
  \begin{equation*}
    \mathfrak{m}(\zeta)=
      \prod\limits_{k\in M}\frac{\zeta-\mu_k}{\zeta-\lambda_k}\,.
  \end{equation*}
\end{proposition}
\begin{proof}
  Consider a sequence $\{M_n\}_{n=1}^\infty$ of subsets of $M$, such
  that $M_n\subset M_{n+1}$ and $\cup_nM_n=M$. From (\ref{eq:levin-herglotz-gen}) and
  (\ref{eq:m-goth-def}) it follows that
  \begin{align}
    \label{eq:m-goth-krein}
    \mathfrak{m}(\zeta)&=C\frac{\zeta-\mu_0}{\zeta-\lambda_0}
    \lim_{n\to\infty}\frac{\displaystyle
  \prod_{\substack{k\in M_n\\k\ne 0}} \left(1-\frac{\zeta}{\eta_k}\right)
  \left(1-\frac{\zeta}{\lambda_k}\right)^{-1}}{\displaystyle
  \prod_{\substack{k\in M_n\\k\ne 0}} \left(1-\frac{\zeta}{\eta_k}\right)
  \left(1-\frac{\zeta}{\mu_k}\right)^{-1}}\notag\\
  &=C\frac{\zeta-\mu_0}{\zeta-\lambda_0}
   \prod_{\substack{k\in M\\k\ne 0}}
\left(1-\frac{\zeta}{\mu_k}\right)
  \left(1-\frac{\zeta}{\lambda_k}\right)^{-1}\,.
  \end{align}
  On the other hand, by
  Proposition~\ref{prop:convergence-eigenvalues}, it holds true that
  \begin{equation}
    \label{eq:infinite-product-equalities}
    \prod_{\substack{k\in M\\k\ne 0}}
\left(1-\frac{\zeta}{\mu_k}\right)
  \left(1-\frac{\zeta}{\lambda_k}\right)^{-1}=
\prod_{\substack{k\in M\\k\ne 0}}\frac{\lambda_k}{\mu_k}
\prod_{\substack{k\in M\\k\ne 0}}\frac{\zeta-\mu_k}{\zeta-\lambda_k}
  \end{equation}
From (\ref{eq:m-asympt}) and (\ref{eq:m-through-m}) it follows that
\begin{equation}
  \label{eq:function-tends-to-1}
\lim_{\substack{\zeta\to\infty \\ \im \zeta\ge\epsilon>0}}
    \mathfrak{m}(\zeta)=1\,.
\end{equation}
Also, on the basis that the second product on the r.\,h.\,s of
(\ref{eq:infinite-product-equalities}) converges uniformly, one has
\begin{equation}
  \label{eq:product-tends-to-1}
    \lim_{\substack{\zeta\to\infty \\
        \im \zeta\ge\epsilon}}
    \prod_{k\in M}
    \frac{\zeta-\mu_k}{\zeta-\lambda_k}=
    \lim_{\substack{\zeta\to\infty \\
        \im \zeta\ge\epsilon}}
    \prod_{k\in M}
    \left(1+\frac{\mu_k-\lambda_k}{\lambda_k-\zeta}\right)=1\,.
\end{equation}
Thus, (\ref{eq:m-goth-krein}), (\ref{eq:infinite-product-equalities}),
(\ref{eq:function-tends-to-1}), and (\ref{eq:product-tends-to-1})
imply that
\begin{equation*}
  C=\prod_{\substack{k\in M\\k\ne 0}}\frac{\mu_k}{\lambda_k}
\end{equation*}
and the proposition is proven.
\end{proof}
\begin{corollary}
  \label{cor:theta}
Fix $g\in\reals\cup\{\infty\}$ and $\theta>0$. Let $J^{(g)}$ have discrete
 spectrum and assume that
  $\sigma(J^{(g)})=\{\lambda_k\}_k$ and
  $\sigma(J^{(g)}(\theta))=\{\mu_k\}_k$, where the sequences have
  been arranged according to Remark~\ref{rem:arrange_sequences}. Then,
  \begin{equation*}
    \theta^2=
      \prod\limits_{k\in M}\frac{\eta-\mu_k}{\eta-\lambda_k}\,.
  \end{equation*}
where $\eta$ is any element of $\sigma(J_1^{(g)})$. Moreover, when
$0\not\in\sigma(J^{(g)})$,
\begin{equation}
  \label{eq:parameter-recovery-zero}
  \theta^2=
      \prod\limits_{k\in M}\frac{\mu_k}{\lambda_k}\,.
\end{equation}
and, if $0\in\sigma(J^{(g)})$,
\begin{equation}
  \label{eq:theta-special-case}
  \theta^2=\frac{1}{\alpha_0-1}\left\{\alpha_0
\prod_{\substack{k\in M\\k\ne 0}}\frac{\mu_k}{\lambda_k} -1\right\}\,,
\end{equation}
where $\alpha_0$ is given in (\ref{eq:def-normalizing}).
\end{corollary}
\begin{proof}
  The first two identities for $\theta^2$ are a straightforward
  consequence of Proposition~\ref{prop:m-goth-actual-form} and
  (\ref{eq:m-through-m}). As regards to (\ref{eq:theta-special-case}),
  note that, from (\ref{eq:m-discrete}), one has
 \begin{equation}
   \label{eq:alpha-residue}
   \alpha_k^{-1}=-\res_{\zeta=\lambda_k}m(\zeta)\,.
 \end{equation}
Thus, according to (\ref{eq:m-through-m}),
\begin{equation}
  \label{eq:0-case-m-goth-0}
  \theta^2-\alpha_0^{-1}(\theta^2-1)=\mathfrak{m}(0)=
\prod_{\substack{k\in M\\k\ne 0}}\frac{\mu_k}{\lambda_k}\,.
\end{equation}
\end{proof}
\begin{remark}
  \label{rem:0-case-m-goth-0}
  Due to (\ref{eq:0-case-m-goth-0}) and the properties of the
  normalizing constants, when $0\in\sigma(J^{(g)})$, one of the following
  inequalities hold depending on the value of $\theta\ne 1$:
 \begin{equation*}
   \theta^2<\mathfrak{m}(0)=
\prod_{\substack{k\in M\\k\ne 0}}\frac{\mu_k}{\lambda_k}<1,\qquad
1<\mathfrak{m}(0)=
\prod_{\substack{k\in M\\k\ne 0}}
\frac{\mu_k}{\lambda_k}<\theta^2\,.
 \end{equation*}
\end{remark}
\begin{theorem}
  \label{prop:reconstruction}
  Fix $g\in\reals\cup\{\infty\}$ and $\theta>0$. Let $J^{(g)}$ have
  discrete spectrum and assume that $0\not\in\sigma(J^{(g)})$. The spectra
  $\sigma(J^{(g)})$, $\sigma(J^{(g)}(\theta))$ ($\theta\ne 1$)
  uniquely determine the Jacobi matrix (\ref{eq:jm-0}), that is the
  operator $J$, the parameter $\theta$ defining the perturbation, and
  the parameter $g$ specifying the self-adjoint extension when $J\ne
  J^*$.
 \end{theorem}
\begin{proof}
  Given the sequences $\sigma(J^{(g)})$ and $\sigma(J^{(g)}(\theta))$,
  one finds the parameter $\theta$ from
  (\ref{eq:parameter-recovery-zero}).
  Proposition~\ref{prop:m-goth-actual-form}
  yields the function $\mathfrak{m}$ and equation
  (\ref{eq:m-through-m}) the Weyl function $m^{(g)}$. According to the
  Preliminaries this function allows to recover the matrix associated
  to the Jacobi operator and the parameter $g$ which determines the
  self-adjoint extension when $J\ne J^*$.
\end{proof}
\begin{theorem}
  \label{prop:reconstruction-2}
  Fix $g\in\reals\cup\{\infty\}$ and $\theta>0$. Let $J^{(g)}$ have
  discrete spectrum and assume that $0\in\sigma(J^{(g)})$. The spectra
  $\sigma(J^{(g)})$, $\sigma(J^{(g)}(\theta))$ ($\theta\ne 1$), together with
  either $q_1$ or $\alpha_0$, uniquely determine the matrix
  associated to $J$, the parameter $\theta$, and the parameter $g$
  when $J\ne
  J^*$. Alternatively, the
  spectra $\sigma(J^{(g)})$, $\sigma(J^{(g)}(\theta))$ and the parameter
  $\theta\ne 1$ uniquely determine the matrix corresponding to $J$ and
  the parameter $g$ when $J$ turns out to be nonself-adjoint.
\end{theorem}
\begin{proof}
  This follows immediately from the proof of the previous theorem,
  taking into account (\ref{eq:theta-special-case}). Note that
  $\theta$ can be determined either by Proposition
  \ref{prop:convergence-eigenvalues} or by the asymptotic formula
  \begin{equation*}
    \mathfrak{m}(\zeta)=1+\frac{q_1(1-\theta^2)}{\zeta}
  +O(\zeta^{-2})\,,
\end{equation*}
as $\zeta\to\infty$ ($\im \zeta\ge \epsilon$, $\epsilon>0$),
obtained by combining (\ref{eq:m-asympt}) and
  (\ref{eq:m-through-m}).
\end{proof}
\begin{remark}
\label{rem:inverse-mass-spring}
Theorems \ref{prop:reconstruction} and \ref{prop:reconstruction-2}
solve the problem of reconstructing the matrix from spectral
data. However, in order to solve the inverse problem for the
mass-spring system, one should also recover the masses and spring
constants from the matrix entries. This is actually not difficult as it is
shown below (cf. \cite[Chap.\,8]{marchenko-new}).

On the basis of (\ref{eq:spring-mass}), one finds the equations
\begin{align*}
  k_{j+1}&=-(k_j+q_jm_j),\\
  m_{j+1}&=\frac{k_{j+1}^2}{m_jb_j^2},
\end{align*}
which allow to find recursively all spring constants and masses of the
system from the first spring constant and mass. Note that, when
the parameters $k_1$ and $m_1$ are given, only the quotient $\frac{k_1}{m_1}$
does not depend on the choice of mass unit. This quotient has a
concrete physical meaning: it equals the squared natural
frequency of the mass $m_1$ attached with the spring
$k_1$ to a fixed support. Thus, it is physically convenient to find a way of expressing
$k_j/m_j$ in terms of $k_1/m_1$. This is achieved by means of the
following continued fraction
\begin{equation*}
  \frac{k_{j+1}}{m_{j+1}}=
  \cfrac{-b_j^2}{q_j
          - \cfrac{b_{j-1}^2}{\cdots q_2
          - \cfrac{b_1^2}{q_1+\frac{k_1}{m_1}}}}\,,
\end{equation*}
which is constructed from $\frac{k_1}{m_1}$ upwards (cf. [17
p.~76]). We remark that, unlike the finite matrix case, here one cannot apply
without substantial changes, the method developed in
\cite[Chap.\,8]{marchenko-new} for determining the set of admissible
values for the  quotient
$\frac{k_1}{m_1}$. Admissible values of $\frac{k_1}{m_1}$ are those
for which $\frac{k_{j+1}}{m_{j+1}}$ is a positive real number for
any $j\in\nats$.
\end{remark}
\section{Necessary and sufficient conditions for the spectra\\
of $J^{(g)}$ and $J^{(g)}(\theta)$}
\label{sec:nec-suf}
The following statement gives an if-and-only-if criterion for two
sequences to be the spectra of $J^{(g)}$ and $J^{(g)}(\theta)$. In the
finite case the interlacing condition given in a) (see below) is
necessary and sufficient \cite{delrio-kudryavtsev},\cite{Ram}.

\begin{theorem}
  \label{prop:necessary-sufficient}
  Given two infinite real sequences $\{\lambda_k\}_k$ and
  $\{\mu_k\}_k$ without finite points of accumulation, such that none
  of them contains the zero, there is a unique positive $\theta$, a
  unique operator $J$, and a unique $g\in\reals\cup\{\infty\}$ if
  $J\ne J^*$, such that $\{\mu_k\}_k$ is the spectrum of
  $J^{(g)}(\theta)$ and $\{\lambda_k\}_k$ is the spectrum of $J^{(g)}$
  if and only if the following conditions are satisfied.
\begin{enumerate}[\ a)]
\item $\{\lambda_k\}_k$ and $\{\mu_k\}_k$ interlace in $\reals_+$,
  $\reals_-$ with one sequence shifted to the right (left) in
  $\reals_+$, ($\reals_-$) with respect to the other one. Thus, the
  sequences can be ordered according to
  Remark~\ref{rem:arrange_sequences}.
    \label{interlace-sufficient}
\item The following series converges
  \begin{equation*}
   \sum_{k\in M}(\mu_k-\lambda_k)
\end{equation*}\label{sum-spectr-sufficient}
By condition \ref{sum-spectr-sufficient}) the products
$\displaystyle\prod_{\substack{k\in M\\k\ne n}}
\frac{\mu_k-\lambda_n}{\lambda_k-\lambda_n}$,
$\displaystyle\prod_{k\in M}
\frac{\mu_k}{\lambda_k}$
 are convergent, so
define
  \begin{equation}
    \label{eq:tau-def-1}
    \tau_n:=
   \frac{(\mu_n-\lambda_n)
\displaystyle\prod_{\substack{k\in M\\k\ne n}}
\frac{\mu_k-\lambda_n}{\lambda_k-\lambda_n}
}{\lambda_n\left(\displaystyle\prod_{k\in M}
\frac{\mu_k}{\lambda_k}-1\right)}\,,
  \quad \forall n\in M\,.
  \end{equation}
\item The sequence $\{\tau_n\}_{n\in M}$
is such that, for $m=0,1,2,\dots$, the series
\begin{equation*}
    \sum_{k\in M}\lambda_k^{2m}\tau_k
    \quad\text{converges.}
\end{equation*}
\label{finite-moments-sufficient}
\item If a sequence of complex numbers $\{\beta_k\}_{k\in M}$
  is such that the series
  \begin{equation*}
    \sum_{k\in
      M}\abs{\beta_k}^2\tau_k
\quad\text{converges}
  \end{equation*}
and, for $m=0,1,2,\dots$,
\begin{equation*}
  \sum_{k\in
    M}\beta_k\lambda_k^m\tau_k=0\,,
\end{equation*}
then $\beta_k=0$ for all $k\in M$.
\label{density-poly-sufficient}
  \end{enumerate}
\end{theorem}
\begin{proof}
  In view of Propositions \ref{prop:interlacing} and
  \ref{prop:convergence-eigenvalues}, for proving the necessity of
  the conditions, it only remains to show that for all $n\in M$,
  $\tau_n=\alpha_n^{-1}$. Indeed
  \emph{\ref{finite-moments-sufficient}}) and
  \emph{\ref{density-poly-sufficient}}) will follow from the fact
  that all moments of the spectral measure
  (\ref{eq:rho-discrete}) exist and that the
  polynomials are dense in $L_2(\reals,\rho)$.

  From (\ref{eq:m-through-m}), (\ref{eq:alpha-residue}), and
  Proposition \ref{prop:m-goth-actual-form} , it follows that
\begin{align*}
  \alpha_n^{-1}&=\frac{1}{\theta^2-1}\lim_{\zeta\to\lambda_n}
  \frac{\lambda_n-\zeta}{\zeta}\mathfrak{m}(\zeta)\\
  &=\frac{\mu_n-\lambda_n}{\lambda_n(\theta^2-1)}
  \prod_{\substack{k\in M\\k\ne n}}
    \frac{\lambda_n-\mu_k}{\lambda_n-\lambda_k}\,.
\end{align*}
Hence, by Corollary \ref{cor:theta}, one verifies that
$\tau_n=\alpha_n^{-1}$.

 We now prove that conditions
  \emph{\ref{interlace-sufficient}}),
  \emph{\ref{sum-spectr-sufficient}}),
  \emph{\ref{finite-moments-sufficient}}), and
  \emph{\ref{density-poly-sufficient}}) are sufficient.

The condition \emph{\ref{interlace-sufficient}}) implies that
\begin{equation*}
  \frac{\lambda_n-\mu_k}{\lambda_n-\lambda_k}>0,\qquad
  \forall k\in M\,,\ k\ne n.
\end{equation*}
On the other hand, by \emph{\ref{sum-spectr-sufficient}}) one can
define the number
\begin{equation}
  \label{eq:kappa-def}
  \kappa=\prod_{k\in M}
\frac{\mu_k}{\lambda_k},
\end{equation}
which is clearly positive and also $\kappa>1$ if
$\abs{\mu_k}>\abs{\lambda_k}$ for all $k\in M$ and $\kappa<1$ if
$\abs{\mu_k}<\abs{\lambda_k}$ for all $k\in M$. Thus,
\begin{equation*}
  \frac{\mu_n-\lambda_n}{\lambda_n(\kappa-1)}>0
  \qquad
  \forall n\in M.
\end{equation*}
Hence, for all $n\in M$, $\tau_n>0$, so define the function
\begin{equation}
  \label{eq:rho-fro-proof}
  \rho(t):=\sum_{\lambda_k<t}\tau_k\,.
\end{equation}
It follows from  \emph{\ref{finite-moments-sufficient}}) that the
moments of the measure corresponding to $\rho$ are finite.

Now, on the basis of \emph{\ref{interlace-sufficient}}) and
  \emph{\ref{sum-spectr-sufficient}}), define the meromorphic functions
\begin{equation*}
  \widetilde{\mathfrak{m}}(\zeta):=
\prod_{k\in M}\frac{\zeta-\mu_k}{\zeta-\lambda_k}
\end{equation*}
and
\begin{equation}
  \label{eq:definition-m-tilde}
  \widetilde{m}(\zeta):=
  \frac{\widetilde{\mathfrak{m}}(\zeta)-
\displaystyle\prod_{k\in M}
\frac{\mu_k}{\lambda_k}}
  {\zeta\left(\displaystyle\prod_{k\in M}
\frac{\mu_k}{\lambda_k}-1\right)}\,.
\end{equation}
Thus, taking into account (\ref{eq:tau-def-1}), one has
\begin{equation}
  \label{eq:residue-tilde}
  \res_{\zeta=\lambda_n}\widetilde{m}(\zeta)=
\left(\prod_{k\in M}
\frac{\mu_k}{\lambda_k}-1\right)^{-1}\lim_{\zeta\to\lambda_n}
\frac{\zeta-\lambda_n}{\zeta}\widetilde{\mathfrak{m}}(\zeta)
=-\tau_n\,.
\end{equation}
In view of what was done earlier,
\begin{equation}
  \label{eq:limit-tilde-one}
  \lim_{\substack{\zeta\to\infty \\ \im \zeta\ge\epsilon>0}}
    \widetilde{\mathfrak{m}}(\zeta)=1\,.
\end{equation}
Therefore,
\begin{equation}
  \label{eq:limit-tilde}
  \lim_{\substack{\zeta\to\infty \\ \im \zeta\ge\epsilon>0}}
    \widetilde{m}(\zeta)=\left(\prod_{k\in M}
\frac{\mu_k}{\lambda_k}-1\right)^{-1}
\lim_{\substack{\zeta\to\infty \\ \im \zeta\ge\epsilon>0}}
\frac{\widetilde{\mathfrak{m}}(\zeta)}
{\zeta}=0
\end{equation}
By (\ref{eq:residue-tilde}) and (\ref{eq:limit-tilde}),
\cite[Chap.\,7,\,Thm.\,2]{MR589888} implies that
\begin{equation}
  \label{eq:m-tilde-as-sum}
  \widetilde{m}(\zeta)=
\sum_{k\in M}\frac{\tau_k}{\lambda_k-\zeta}\,.
\end{equation}
On the other hand, using (\ref{eq:limit-tilde-one}), one obtains
\begin{equation*}
  \lim_{\substack{\zeta\to\infty \\ \im \zeta\ge\epsilon>0}}
    \zeta\widetilde{m}(\zeta)=\left(\prod_{k\in M}
\frac{\mu_k}{\lambda_k}-1\right)^{-1}
\lim_{\substack{\zeta\to\infty \\ \im \zeta\ge\epsilon>0}}
\left(\widetilde{\mathfrak{m}}(\zeta)
-\prod_{k\in M}\frac{\mu_k}{\lambda_k}\right)=-1\,.
\end{equation*}
But
\begin{equation*}
  \lim_{\substack{\zeta\to\infty \\ \im \zeta\ge\epsilon>0}}
    \zeta\widetilde{m}(\zeta)=-\sum_{k\in M}\tau_k\,,
\end{equation*}
so it has been proven that, for the function given in
(\ref{eq:rho-fro-proof}),
\begin{equation*}
  \int_\reals d\rho(t)=1\,.
\end{equation*}
Thus the measure corresponding to $\rho$ is appropriately normalized
and all the moments exist, so in $L_2(\reals,\rho)$ apply the
Gram-Schmidt procedure of orthonormalization to the sequence
$\{t_k\}_{k=0}^\infty$ to obtain a Jacobi matrix as was explained in
the Preliminaries. Denote by $J$ the operator whose matrix
representation is the obtained matrix (cf. \cite[Sec. 47]{MR1255973}).
Now, depending on the sequence of moments, $J$ is self-adjoint or
not. If $J=J^*$, the function $\rho$ is the resolution of the identity
of $J$, while if $J\ne J^*$, $\rho$ corresponds to the resolution of
the identity of a self-adjoint extension of $J$. This is a
consequence of condition \emph{\ref{density-poly-sufficient}}) since
it means that the polynomials are dense in $L_2(\reals,\rho)$
\cite[Prop.\,4.15]{MR1627806}.

Finally, denote by $J^{(g)}$ the self-adjoint extension of $J$
corresponding to $\rho$ and consider the operator $J^{(g)}(\theta)$
obtained from $J^{(g)}$ as indicated in the Preliminaries with
$\theta$ given by (\ref{eq:parameter-recovery-zero}). By construction
the sequence $\{\lambda_k\}_{k\in M}$ is the spectrum of
$J^{(g)}$. For the proof to be complete it only remains to show that
$\{\mu_k\}_{k\in M}$ is the spectrum of $J^{(g)}(\theta)$. For the
function given in (\ref{eq:m-goth-def}), taking into account
(\ref{eq:m-through-m}) and (\ref{eq:m-discrete}), one has
\begin{equation*}
  \mathfrak{m}(\zeta)=\theta^2+\zeta\left(\theta^2-1\right)
  \sum_{k\in M}\frac{1}{\alpha_k(\lambda_k-\zeta)}\,.
\end{equation*}
On the other hand, from (\ref{eq:definition-m-tilde}) and
(\ref{eq:m-tilde-as-sum}), it follows that
\begin{equation*}
  \widetilde{\mathfrak{m}}(\zeta)=\theta^2+\zeta\left(\theta^2-1\right)
  \sum_{k\in M}\frac{\tau_k}{\lambda_k-\zeta}\,.
\end{equation*}
But we have already proven that $\alpha_k^{-1}=\tau_k$ for $k\in
M$. Thus $\mathfrak{m}=\widetilde{\mathfrak{m}}$, meaning that the
zeros of $\mathfrak{m}$ are given by the sequence $\{\mu_k\}_{k\in
  M}$.
\end{proof}
\begin{theorem}
  \label{prop:necessary-sufficient-zero}
  Let $\{\lambda_k\}_k$ and $\{\mu_k\}_k$ be two infinite real
  sequences without finite points of accumulation, such that each of
  them contains exactly one element equal zero, and consider any
  positive real number $\theta\ne 1$. There exists a unique operator
  $J$, and a unique $g\in\reals\cup\{\infty\}$ if $J\ne J^*$, such
  that $\{\mu_k\}_k$ is the spectrum of $J^{(g)}(\theta)$ and
  $\{\lambda_k\}_k$ is the spectrum of $J^{(g)}$ if and only if the
  conditions \ref{interlace-sufficient}),
  \ref{sum-spectr-sufficient}), \ref{finite-moments-sufficient}), and
  \ref{density-poly-sufficient}) hold with
  \begin{align*}
        \tau_n&:=
        \frac{\mu_n-\lambda_n}{\lambda_n\left(\theta^2-1\right)}
\prod_{\substack{k\in M\\k\ne n}}
\frac{\mu_k-\lambda_n}{\lambda_k-\lambda_n}\,,
  \quad n\in M\,, n\ne 0\,, \\
       \tau_0&:=(\theta^2-1)^{-1}\left(\theta^2-
\prod_{\substack{k\in M\\k\ne 0}}
\frac{\mu_k}{\lambda_k}\right)\,,
  \end{align*}
where
\begin{equation}
  \label{eq:bound-on-theta}
  \theta^2
  \begin{cases}
    <\displaystyle\prod_{\substack{k\in M\\k\ne 0}}\frac{\mu_k}{\lambda_k} &
    \text{if}\ \{\mu_k\}_k\ \text{is shifted to the
      left in}\
    \reals_+\ \text{w.r.t.}\ \{\lambda_k\}_k\,,\\
    >\displaystyle\prod_{\substack{k\in M\\k\ne 0}}\frac{\mu_k}{\lambda_k} &
    \text{otherwise}.
  \end{cases}
\end{equation}
\end{theorem}
\begin{proof}
  The proof is analogous to the proof of
  Theorem~\ref{prop:necessary-sufficient}. Recall that by our
  convention for enumerating the sequences $\lambda_0=\mu_0=0$. Thus,
  for proving the necessity of the conditions
  \emph{\ref{interlace-sufficient}})--\emph{\ref{density-poly-sufficient}}),
  one only should verify that $\tau_0=a_0^{-1}$ and
  (\ref{eq:bound-on-theta}) holds. This is immediate in view of
  (\ref{eq:0-case-m-goth-0}) and Remark~\ref{rem:0-case-m-goth-0}. The
  sufficiency of the conditions is established as in the proof of
  Theorem~\ref{prop:necessary-sufficient}. Here, one substitutes
  (\ref{eq:kappa-def}) by
  \begin{equation*}
      \kappa=\prod_{\substack{k\in M\\k\ne 0}}
\frac{\mu_k}{\lambda_k}
  \end{equation*}
and (\ref{eq:definition-m-tilde}) by
\begin{equation*}
    \widetilde{m}(\zeta):=
  \frac{\widetilde{\mathfrak{m}}(\zeta)-\theta^2}
  {\zeta\left(\theta^2-1\right)}\,,\qquad \zeta\ne 0\,.
\end{equation*}
Then, one verifies that $
\res_{\zeta=\lambda_n}\widetilde{m}(\zeta)=-\tau_n$ for all $n\in M$
and $\sum_{k\in M}\tau_k=1$. Note that (\ref{eq:bound-on-theta})
guarantees that $\tau_n>0$ for all $n\in M$. The rest of the proof
repeats that of  Theorem~\ref{prop:necessary-sufficient} taking
into account that now the zeros of $\mathfrak{m}$ are given by
$\{\mu_k\}_{k\in M}\setminus\{0\}$.
\end{proof}
\begin{theorem}
  \label{prop:necessary-sufficient-aa}
  Given two infinite real sequences $\{\lambda_k\}_k$ and
  $\{\mu_k\}_k$ without finite points of accumulation, such that none
  of them contains the zero, there is a unique positive $\theta$ and a
  unique operator $J=J^*$ such that $\{\mu_k\}_k$ is the spectrum of
  $J^{(g)}(\theta)$ and $\{\lambda_k\}_k$ is the spectrum of $J$
  if and only if conditions \ref{interlace-sufficient}),
  \ref{sum-spectr-sufficient}), \ref{finite-moments-sufficient}),
  together with
  \begin{equation*}
  \text{d')}\qquad\qquad\qquad
  \lim_{n\to\infty}
    \frac{\det\begin{pmatrix}
      s_0 & s_1 & \cdots & s_n\\[1mm]
      s_1 & s_2 & \cdots & s_{n+1}\\[1mm]
      \hdotsfor[2]{4}\\[1mm]
      s_n & s_{n+1} & \cdots & s_{2n}
    \end{pmatrix}}
    {\det\begin{pmatrix}
      s_4 & s_5 & \cdots & s_{n+2}\\[1mm]
      s_5 & s_6 & \cdots & s_{n+3}\\[1mm]
      \hdotsfor[2]{4}\\[1mm]
      s_{n+2} & s_{n+3} & \cdots & s_{2n}
    \end{pmatrix}}=0\,,\qquad\qquad\qquad
  \end{equation*}
where $s_n:=\sum_{k\in M}\lambda_k^n\tau_k$ for $n$ in $\nats\cup\{0\}$
are fulfilled. Note that by our convention on the notation
$J^{(g)}(\theta)$ is a non-singular finite-rank perturbation of $J$
which does not depend on $g$.
\end{theorem}
\begin{proof}
  We again repeat the reasoning of the proof of
  Theorem~\ref{prop:necessary-sufficient}. Clearly, $s_n$
  ($n\in\nats\cup\{0\}$) are the numbers given in
  (\ref{eq:moments}). Thus, on the basis of Hamburger criterion (see
  \cite[Addenda\,2,\,Sec.\,9]{MR0184042}),
  \emph{d')} holds when $J=J^*$. For the sufficiency, note that, due
  to \cite[Addenda\,2,\,Sec.\,9]{MR0184042},
  \emph{d')} implies that the
  measure corresponding to the function given in
  (\ref{eq:rho-fro-proof}) is the unique solution of the moment
  problem, so $J=J^*$ and \emph{d)} is not needed.
\end{proof}
\begin{remark}
  \label{last}
Admittedly,  \emph{d')} is not easy to check, however it allows to give
necessary and sufficient conditions in the self-adjoint case. Note
that one can also give the analogous self-adjoint version of
Theorem~\ref{prop:necessary-sufficient-zero} by
substituting condition \emph{d)} for \emph{d')}.
\end{remark}
\begin{acknowledgments}
The authors thank the referee whose comments have led to an improved
presentation of this work.
\end{acknowledgments}

\end{document}